\documentclass[11pt,fleqn,twocolumn]{article}
\usepackage{latexsym,amssymb,amsmath,theorem,epsfig}
\usepackage[hyperindex]{hyperref}

\usepackage{bm,amsfonts}
\usepackage{algorithmic}
\usepackage{dcolumn}
\usepackage[affil-it]{authblk}
\usepackage{blindtext}
\usepackage{abstract}
\newtheorem{theorem}{Theorem}[section]
\newtheorem{algorithm}[theorem]{Algorithm}

{\theorembodyfont{\rmfamily} }

\newtheorem{proof}[theorem]{Proof}

{\theorembodyfont{\rmfamily} }

\newtheorem{thm}{Theorem}[section]

\newtheorem{prop}[thm]{Proposition}

\begin{document}

\title{\LARGE\sf Synchronization of cyclic power grids: equilibria and stability of the synchronous state
\footnote{Submitted to Chaos in 29 July 2016.}}
\author
{
{\large\sf Kaihua Xi\footnote{Email:~K.Xi@tudelft.nl},  Johan.L.A. Dubbeldam\footnote{Email:~j.l.a.dubbeldam@tudelft.nl}, 
Haixiang Lin\footnote{Email:~H.X.Lin@tudelft.nl}}\\ 
{\small\sf Delft Institute of Applied Mathematics, Delft University of Technology, }\\
{\small\sf 2628 CD, Delft, The Netherlands}\\
}
\date{\vspace{-5ex}}

\twocolumn[
  \begin{@twocolumnfalse}
    \maketitle
      \textbf{Abstract}
      
      ~~Synchronization is essential for proper functioning of the power grid.
We investigate the synchronous state
and its stability for a network with a cyclic topology and with the
evolution of the states satisfying the swing equations. We calculate the number of stable 
equilibria and investigate both the linear and nonlinear stability of the
synchronous state. The linear stability analysis shows that the
stability of the state, determined by the smallest nonzero eigenvalue,
is inversely proportional to the size of the network. The nonlinear stability, 
which we calculated by comparing the
potential energy of the type-1 saddles with that of the stable synchronous
state, depends on the network size ($N$) in a more complicated
fashion. In particular we find that when the generators and consumers
are evenly distributed in an alternating way, the energy barrier,
preventing loss of synchronization approaches a constant value. For a
heterogeneous distribution of generators and consumers, the energy
barrier will
decrease with $N$. The more heterogeneous the distribution is, the stronger
the energy barrier depends on $N$.
Finally, we found that by comparing situations with equal line loads in
cyclic and tree networks, tree networks
exhibit reduced stability. This difference disappears in the limit of
$N\to\infty$. This finding corroborates previous results reported in the
literature
and suggests that cyclic (sub)networks may be applied to enhance power
transfer while maintaining stable synchronous operation.\\
\quad \\
 \end{@twocolumnfalse}
]
\saythanks

\section{Introduction}\label{sec:intro}

 The electrical power grid is a fundamental infrastructure in today's society.
Its enormous complexity makes it one of the most complex systems ever engineered by humans.
The highly interconnected structure of power grid delivers power over a long distance.
However, it also propagates local failures into the global network causing cascading failures.
 Due to careful control and management,
it has been operating for decades, mostly with great reliability. However, massive blackouts still occur,
such as, for example, the power failures in 2003 in the
Northern US and Canada and more recently a huge outage occurred in Turkey (70 million people affected) in March 2015.

The current transition to a more distributed generation of energy by renewable sources, which are inherently more
 prone to fluctuations, poses even greater
challenges to the functioning of the power grid. As the contribution of renewable energy to the total power 
being generated is surging, it becomes more challenging to keep the network stable against large disturbances.
In particular, it is essential that  power generators remain \emph{synchronized}. 
The objective of this paper is to study the influence of the distribution of power generation and consumption on the 
synchronization. We find that the heterogeneity of power generation and consumption decreases both the linear 
stability and the nonlinear stability. We show that large size cyclic power grids are more sensitive to 
the heterogeneity. In addition, a finding suggests that a line in a tree network loses synchronization more easily than
a line carrying the same amount of power in a ring network. The stability of the synchronous state 
can be improved by forming small cycles in the network. This finding may help optimize the power flow and 
design the topology of future power grids.

The significance of stable operation of the power grid was acknowledged long ago and has led to general stability studies for the power 
grid using direct methods\cite{Varaiya,chiang1,chiang2,Liu_chosest,Chiang3,Chiang4}. 
More recently
conditions for linear stability against small size disturbances were derived by Motter \emph{et al.}~\cite{motter} 
using the master stability 
formalism of Pecora and Carroll \cite{pecora}.
Complementary work on large disturbances was described in \cite{menck2}. In this work the concept of \emph{basin stability}
was applied to estimate the basin of attraction of the synchronous state. In particular
it was demonstrated that so-called dead-ends in a network may greatly reduce stability of the synchronous state.


The primary interest of this paper is to study the influence of the distribution of generators and
consumers on synchronization.
Starting from the commonly used second-order swing equations, we reduce our model
 to a system of first-order differential equations using the techniques developed by Varaiya, Chiang \emph{et al.}\cite{Varaiya,book}
 to find  stability regions for synchronous operation of electric grids after a contingency. 
 After this reduction we are left with a first-order Kuramoto model with nearest neighbor coupling.

For the case of a ring network with homogeneous distribution of
generation and power consumption, we obtain analytical results which generalize earlier work of
De Ville~\cite{Lee}. In particular we derive analytical expressions for the stable equilibria and calculate their number. 
Furthermore, we investigate the more general case with
random distributions of generators and consumers numerically. To this end we develop a novel
algorithm that allows fast determination of the stable equilibria, as well as the saddle points in the system.

Subsequently, the stability of the equilibria is studied both using linearization techniques for linear stability and
direct (energy) methods for determining the nonlinear stability.
By comparing our stability results for different network sizes we show that the linear stability
properties differ greatly from those obtained by direct methods when the system size increases.
More specifically, the linear stability, measured by the first nonzero eigenvalue approximately scales
with the inverse of the number of nodes ($N$) as $1/N$.  This is in contrast to  the nonlinear stability result, which shows
that the potential energy barrier that prevents the synchronous state from lo osing stability approaches a nonzero value
for $N\to\infty$. For large size cyclic power grids, small perturbation on the power supply or consumer may lead
to desynchronization.
 Moreover, comparison of a ring topology with a tree topology,
reveals enhanced stability for the ring configuration.
This result suggests that the finding that dead-ends or dead-trees diminish stability by 
Menck \emph{et al.} \cite{menck2} can be interpreted as a special case of the more general fact 
that tree-like connection desynchronize easier than ring-like connection. 

This paper is organized as follows. In section \ref{sec:model} we define the model. In section \ref{sec:equil} 
we calculate the (number of) stable equilibria of 
cyclic networks and study the existence of the synchronized state. 
We next analyze the linear stability of the synchronous states in section \ref{sec:linear}. 
The proofs of the statements
in section \ref{sec:linear} are presented in the Appendix. In section \ref{Sec:nonlinear} we consider the nonlinear stability of the
synchronous state, which is measured by the potential energy difference between the saddles and the 
stable equilibrium. 
Finally we conclude with a summary of our results in section \ref{sec.conclusion}.

\section{Introduction of the model}
\label{sec:model}

The model that we use in this paper is commonly known as the swing equation model and has been
derived in a number of books and papers; see for example \cite{menck2,timme1,timme2,bergen}. This model is sometimes also referred to as a 
second-order Kuramoto model or a Kuramoto model with inertia \cite{filatrella}. The swing equation model is defined by the 
following differential equations
\begin{align}
 \frac{d^2\delta_i}{dt^2}+\alpha\frac{d\delta_i}{dt}+K\sum_{j}A_{ij}\sin(\delta_i-\delta_j)=P_i,
\label{eq:swing}
 \end{align}
where the summation is over all $N$ nodes in the network.
In Eq.~(\ref{eq:swing})  $\delta_i$ is the phase of the $i-$th generator/load and $P_i$ is
the power that is generated ($P_i>0$) or consumed ($P_i<0$) at node $i$ and $\alpha$ is the damping parameter  that
we take equal for all nodes. The link  or coupling strength
is denoted by ($K$) and $A_{ij}$ is the coefficient in the adjacency matrix of the network.

When we consider the case of a ring network, Eq.~(\ref{eq:swing}) reduce to
\begin{align}
 \hspace{-25pt}\frac{d^2\delta_i}{dt^2}+\alpha\frac{d\delta_i}{dt}+K[\sin(\delta_i-\delta_{i+1})+\sin(\delta_i-\delta_{i-1})]=P_i,
\label{eq:rswing}
 \end{align}
 with $i=1,2,\dots,N$. In writing Eq.~(\ref{eq:rswing}) we assumed that $\delta_{N+i}=\delta_i$.
We  usually rewrite the second-order differential equations  as the first-order system
\begin{align}
\hspace{-15pt}
\dot{\delta}_i&=\omega_i ,\nonumber\\
\hspace{-15pt}
\dot{\omega}_i&=P_i-\alpha\omega_i-K [\sin(\delta_i-\delta_{i+1})+\sin(\delta_i-\delta_{i-1})]. \label{eq:stelsel}
\end{align}
Note that since the total consumption must equal the total amount of power being generated in equilibrium,
 synchronous operation of the system implies that $$\sum_{i=1}^N P_i=0.$$

Let us assume that node $i$ is a generator, then $P_i>0$. The term  ${\cal{L}}_{i,i-1}=K\sin(\delta_i-\delta_{i-1})$ in Eq.~(\ref{eq:stelsel}) 
then corresponds to the power that is transported from node $i$ to node $i-1$ and we will refer to the quantity ${\cal{L}}_{i,i-1}$ as 
the \emph{line load} of the line between node $i$ and $i-1$. In a
similar way ${\cal{L}}_{i,i+1}=K\sin(\delta_i-\delta_{i+1})$ is the line load of the link connecting node $i$ and $i+1$.
For the case $i$ is a consumer node, a similar interpretation can be given.  

In this paper we will focus on two different models. The first model is the model in 
which power $P$ is generated at the odd nodes and $-P$ is consumed at the even nodes, which we can capture as 
\begin{equation}
P_i=(-1)^{i+1}P.
\label{eq:ALT}
\end{equation}
We will refer to this model as the \emph{homogeneous model}. 

In the second model we break the symmetry and allow variations in the power generated and consumed at each node, but in
such a way that the net total generated and consumed power vanishes ($\sum_{i=1}^N P_i=0$). This can be accomplished by
\begin{align}
&P_i=(-1)^{i+1}P+\xi_i,~~~~~~~ i=1,2,\dots N-1.\nonumber\\
&\xi_i\in N(0,\sigma),~~{\rm and}~~\sum_{i=1}^N {P_i}=0.
\label{eq:pres}
\end{align}
Here $N(0,\sigma)$ is the normal distribution with standard deviation $\sigma>0$ and mean $0$ and $\xi_{i}$ is a random number.
Eq.~(\ref{eq:pres}) expresses that the new distribution of generators and consumers is
obtained by a Gaussian perturbation of the homogeneous model that we started with.
This model with a heterogeneous distribution of generated and consumed power will be referred to as the \emph{heterogeneous model} and
the degree of heterogeneity of $P_{i}$ is measured by $\sigma$. 
 
To investigate the linear stability of the synchronous (equilibrium) state, Eqs.~(\ref{eq:stelsel}) are linearized
around an equilibrium state $(\delta_i^s,0),~~i=1,\dots,N$. Using vector notation
${\bm{ \delta}}=(\delta_1,\dots,\delta_N)^T$ and  ${\bm{ \omega}}=(\omega_1,\dots,\omega_N)^T$, the linearized dynamics
is given by the matrix differential equation
\begin{align}
\left(\begin{array}{c} {\bm{\dot{\delta}}} \\ {\bm{ \dot{\omega}}}\end{array}\right)&=\left(\begin{array}{cc} 0 & I_N\\
L & -\alpha\end{array}\right) \left(\begin{array}{c} {\bm\delta} \\{\bm \omega}\end{array}\right)= J \left(\begin{array}{c} {\bm{\delta}} \\{\bm{ \omega}}\end{array}\right),
\label{eq:matrix}
\end{align}
with $L$ the (negative) Laplacian matrix defined by
\begin{align}
L_{i,i-1}&=K \cos(\delta_i-\delta_{i-1}),\nonumber\\
L_{i,i+1}&=K \cos(\delta_i-\delta_{i+1}),\nonumber\\
L_{i,i}&=-K [\cos(\delta_i-\delta_{i-1})+ \cos(\delta_i-\delta_{i+1})].
\end{align}
The eigenvalues of $L$, denoted by $\lambda_i$, are related to the eigenvalues of $J$, denoted by $\mu_{i}$, according to the following equation
\begin{align}
\mu_{i\pm}=-\frac{\alpha}{2}{\pm}\frac{1}{2}\sqrt{\alpha^2+4\lambda_i},~~~~ i=1,2,\dots,N-1.
\label{eq:mu}
\end{align}
These $2N-2$ eigenvalues are supplemented by two eigenvalues $0$; one corresponding to a uniform frequency shift, 
the other to a uniform phase shift. For $\alpha>0$, the real part of $\mu_{i\pm}$ is negative if $\lambda_{i}<0$. 
The \emph{ type-$j$ equilibria} are defined as the ones whose Jacobian matrix $J$ have $j$ eigenvalues with a positive real part.

\section{The equilibria of ring networks}
\label{sec:equil}
In this section, we study a ring network consisting of an even number of nodes ($N$) with $N/2$ generators and $N/2$ consumers,
which are connected alternatingly by $N$ links as shown in Fig.~\ref{fig:fig1}. The phase differences between neighbors are
\begin{align}
\nonumber
\hspace{-25pt}
\theta_1 \equiv\delta_{1}-\delta_{N}  (\rm{mod}~~2\pi),~\theta_{i+1} \equiv \delta_{i+1}-\delta_{i} (\rm{mod}~~2\pi).
\end{align}

\begin{figure}[ht]
 \centering
  \includegraphics[scale=0.4]{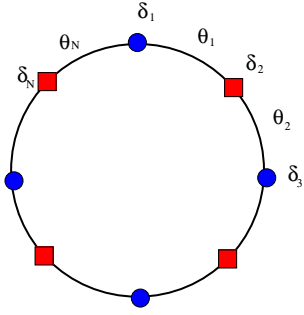}
\caption{{\small A ring network with alternating consumer and generator nodes. Circle nodes are generators and square nodes are consumers.}}
 \label{fig:fig1}
\end{figure}

To find the equilibrium points we set $(\delta_i,\dot{\delta}_i)=(\delta^s,0)$ in Eqs.~(\ref{eq:stelsel})
from which we find the following equations for $\theta_i$
\begin{align}
\sin \theta_{i} -\sin \theta_{i+1}&=P_{i}/K,~~~~~~~~~~ i=1,\cdots, N
\label{eq:equil}
\end{align}

Because all phase differences $\theta_i$ are restricted to
one period of $\theta_{i}$, the following additional requirement holds
{\small 
\begin{align}
\hspace{-25pt}\sum_{i=1}^{N} \theta_{i}=2 m\pi,~~~ m\in~\{-\lfloor N/2\rfloor  ,\dots,-1,0,1,\dots,\lfloor N/2\rfloor\},
\label{eq:additional_require}
\end{align}
}
where $\lfloor N/2 \rfloor$ denotes the floor value of $N/2$, that is, the largest integer 
value which is smaller than or equal to $N/2$. Each equilibrium corresponds to a synchronous state whose stability we wish to determine.
We first calculate the number of stable equilibria of the homogeneous model. 
Note that these equilibria correspond to phase-locked solutions of a first-order Kuramoto model that was explored by 
De Ville~\cite{Lee} for the case $P=0$.



\subsection{The equilibria of homogeneous model}

In this subsection, the number of stable equilibria is determined by solving the nonlinear system analytically.
Our approach is similar to that of Ochab and G\'ora \cite{ochab}. In the homogeneous model 
 we have $P_i=(-1)^{i+1} P$, so that
consumers correspond to even nodes and generators to odd nodes in the network.

It can easily be shown that the relative phase differences $\theta_i$ for all even values
are either the same, that is $\theta_{2i}=\theta_{2}$, or satisfy $\theta_{2i}=\pi-\theta_{2}$ for all $i=1,2,\cdots N/2$.
Similarly, for odd values of $i$ the $\theta_i$ are all
equal to $\theta_1$ or to $\pi-\theta_1$. In this subsection, we consider 
the case $\theta_{i}\in[-\pi/2,\pi/2]$. So we need only consider equation
(\ref{eq:equil}) for a single value of $i$, which we take to be $i=1$. 
The equations for homogeneous model read
\begin{subequations}
\begin{align}
\sin \theta_1-\sin \theta_2&=P/K \label{a},\\
\theta_1+\theta_2&= \frac{4 \pi m}{N},\label{b}
\end{align}
\end{subequations}
We substitute  the value for $\theta_2=\frac{4\pi m}{N}-\theta_1$ in Eq.~(\ref{a}) and solve for $\theta_1$.
From that we can find that the equilibrium values for $\theta_1$ and $\theta_2$, which are given by
\begin{subequations}\label{eqs:theta12}
\begin{align}
\theta_1&=\arcsin\left[\frac{P}{2K\cos \frac{2m \pi}{N}}\right]+\frac{2\pi m}{N},\label{t1a}\\
\theta_2&=-\arcsin\left[\frac{P}{2K\cos \frac{2m \pi}{N}}\right]+\frac{2\pi m}{N},\label{t1b}
\end{align}
\end{subequations}
where $m \in \{ \lfloor N/2\rfloor,\dots,-1,0,1,\dots,\lfloor N/2\rfloor \}$.
 
Here we remark that in order to find a solution the condition $\frac{P}{2K\cos \frac{m \pi}{N}}\,{\leq}\,1$ needs to be satisfied. 
We can use this requirement to find the total number of stable equilibria. Later we show that
only the equilibria with $\theta_1,\theta_2 \in [-\pi/2,\pi/2]$ are stable. Imposing this additional requirement
results in a total number of stable equilibria given by
\begin{align}
N_S=1+2 \lfloor \frac{N}{2\pi} \arccos\left(\sqrt{\frac{P}{2K}}\right)\rfloor.
\label{eq:numbstab}
\end{align}
Details of this calculation can be found in Appendix A. In Fig.~\ref{fig:Num_eqs}(a), we show the total number of stable
equilibria $N_S$ as a function of $P/K$.
It can be clearly seen from this figure that the total number of stable equilibria decreases with $P/K$ and 
reaches 0 when $P/K=2$.

Next we will calculate the number of stable equilibria of the cyclic power system with heterogeneous distribution of
generation and consumption. 

\subsection{The equilibria of heterogeneous model}
To investigate the effect of the distribution of $P_i$,
we performed Monte Carlo (MC) simulations of the system (\ref{eq:stelsel}) with the
distribution of $P_i$ given by Eq.~(\ref{eq:pres}). 

All the equilibria of small size power systems can be found 
using a software package Bertini \cite{mehta1,mehta2,Bertini}. However,
we perform numerical calculations using the algorithm of Appendix D since the size of the networks is relatively large. 
The algorithm amounts to
finding all solutions for $\beta$ of
\[
\sum_{i=1}^N a_i \arcsin\left(\sum_{j=1}^i {P_j}/K+\beta\right)= m \pi,
\]
where $a_i=1$ when the phase difference $\theta_i\in[-\pi/2,\pi/2]$ and $a_i=-1$ if $\theta_i\in[\pi/2,3\pi/2]$.
For details and bounds on the values of $m$ we refer to Appendix D.
Since the number of equilibria is known to increase at least exponentially with $N$ \cite{luxem,baillieul}, it is not
feasible to find all equilibria for large networks. Therefore, we developed an algorithm based on
a recent theoretical paper of Bronski and De Ville \cite{bronski} for finding all equilibria of type-$j$.
Details about the algorithm can be found in Appendix D. We are particularly interested in type-1 equilibria,
as a union of the stable manifolds of these equilibria can be used to approximate the basin of stability of the stable 
equilibria. Our algorithm is capable to find type-1 equilibria at a computational cost of ${\rm O}\left(N^3\right)$ 
and hence can be applied to rather large networks. We remark that this algorithm might be extended to more 
general networks. Employing this algorithm, 
the number of stable equilibria are investigated as follows. 

\begin{figure}[ht]
\centering
\includegraphics[scale=1.1]{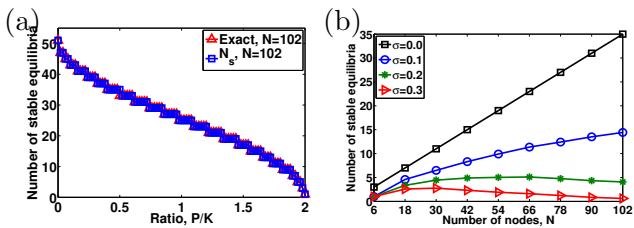}
\caption{{\small (a) The number of stable equilibria according to Eq.~(\ref{eq:numbstab}) compared to the
numerically calculated  number of stable equilibria.
(b) The number of stable equilibria as a function of $N$, $P/K=0.5$. With larger $\sigma$,
it becomes more difficult for the power system to synchronize to a stable state.
} }
\label{fig:Num_eqs}
\end{figure}


Our algorithm was applied to networks with $P/K=0.5$ described by the heterogeneous model. In the simulations we
average over 1000 independent runs  for each value of $\sigma$. In Fig.~\ref{fig:Num_eqs}(b) the number of stable equilibria
is plotted as a function of the number of nodes that we vary from $N=6$ to $N=102$, for 4 different values
of $\sigma=0$, $0.1$, $0.2$, $0.3$. It can clearly be seen that for $\sigma=0.2$  or $\sigma=0.3$ the number of stable equilibria
attains a maximum value and then decreases. The same behavior is also expected for $\sigma=0.1$. However, the
decrease is expected to set in at larger values of $N$, hence such behavior cannot be observed from Fig.~\ref{fig:Num_eqs}(b).
The occurrence of a maximum for nonzero $\sigma$ can be understood as follows.
If the number of nodes increases, the probability that a  phase difference between two nodes exceeds $\pi/2$ also
increases. Even though for moderately large $N$ ($1<N<50$) the fact that more equilibria can be found increases
linearly with $N$, as was shown in Eq.~(\ref{eq:numbstab}), this increase is much smaller than the decrease caused by the
arising probability of phase differences beyond $\pi/2$. This may be explained by the fact that in larger networks
the probability to form clusters in which neighboring nodes $i$ and $i+1$ have large $\Delta P=|P_i-P_{i+1}|$
increases more rapidly than linearly with $N$. As a larger $\Delta P$ is associated with larger phase differences, 
such clusters with large fluctuations in $\Delta P$ between its members are likely to result in asynchronous behavior. 
This finding is in agreement with the well-known result that no synchronous states
exist for an infinite Kuramoto network; see also \cite{strogatz}.

Note that for certain network distribution of $P_{i}$, equilibria can be found
with at least one phase difference exceeding $\frac{\pi}{2}$, but nevertheless being stable. This is
in accordance with the graph theoretical result of Bronski and De Ville~\cite{bronski} and
numerical findings of Metha \emph{et al.}~\cite{mehta1}.


\section{Linear stability of equilibria}
\label{sec:linear}

To determine the stability of the equilibria, the eigenvalues of the matrix
corresponding to the system of second-order differential equations are required.
These can be calculated analytically for single generator coupled to an infinite bus system for any value
of damping parameter $\alpha$, in which case the system is described by a single second-order differential equation.
Such an approach was also taken by Rohden \emph{et al.} \cite{timme1,timme2}.

The eigenvalues of the linearized system Eq.~(\ref{eq:matrix}) can be explained in forms of the eigenvalues of $L$ as shown 
in Eq.~(\ref{eq:mu}). 
For positive $\alpha$, a positive eigenvalue $\lambda_i$ of $L$ results in a corresponding 
eigenvalue $\mu_i$ with positive real part \cite{zab2}. 
So the stability of equilibrium is determined by the eigenvalues of $L$.

The equilibrium with all eigenvalues of $L$ negative and damping parameter $\alpha$ positive 
is most interesting for power grids. We find that in this case all $N-1$ pairs 
of eigenvalues Eq.~(\ref{eq:mu}) are complex valued with negative real part. Hence the system is stable in this case. 
The most stable situation arises when the damping coefficient $\alpha$ is tuned to the optimal 
value $\alpha_{\rm{opt}}$ described by Motter \emph{et al.} \cite {motter}: $\alpha_{\rm{opt}}=2\sqrt{-\lambda_1}$, 
where $\lambda_1$ is the least negative eigenvalue of $L$, in that
case $\mu_1=-\sqrt{-\lambda_1}$.
So the linear stability is governed by the
eigenvalues of $L$. We will therefore further investigate the eigenvalues of $L$
for ring networks in this section.



The entries of the matrix $L$ that arises after linearization around the synchronized
state $({\bm \delta}^s,{\bf 0})$ are easily calculated and from that we find that $L$ is the 
following Laplacian matrix
{\tiny{
\begin{align}
\hspace{-25pt}
L=\left(\begin{array}{cccccc}
-c_{2}-c_{1} & c_{2}         & 0      & \cdots  &0                & c_{1} \\
c_{2}        &  -c_{2}-c_{3} & c_{3}  & 0       &\cdots           & 0 \\
0            & \ddots        & \ddots & \ddots  &\ddots           & 0\\
0            & \cdots        & 0      &c_{N-2}  &-c_{N-2}-c_{N-1} &c_{N-1}\\
c_{1}        &0              & \cdots & 0       & c_{N-1}         & -c _{1} -c_{N-1}
\end{array} \right),
\label{eq:defL}
\end{align}
}
}
where $c_{i}=K\cos(\delta_{i}-\delta_{i-1})=K\cos{\theta_{i}}$. As matrix $L$ is a (symmetric)
Laplacian matrix
with zero-sum rows, $\lambda=0$ is an eigenvalue. This reflects a symmetry in the system:
if all phases are shifted by the same amount $\beta$, the system of differential
equations remains invariant. It is well known that when all entries
$c_{i}>0$, $L$ is negative definite, hence all eigenvalues are non-positive which implies stable
equilibria when the phase differences $|\delta_i-\delta_{i-1}|\,\leq\,\pi/2~~~(\rm{mod}~ 2\pi)$, for all $i=1,\dots,N$.

\subsection{The linear stability of homogeneous model}

For the configuration with a homogeneous distribution of power generation and consumption 
we can derive a theorem which shows that type-1 saddle points, which are saddle points with one unstable 
eigen direction, appear if
a single phase difference between two nodes has negative cosine value. Saddles with more
unstable directions result when more phase differences have negative cosine value. In the following, 
we write a phase difference 
exceeds $\pi/2~(\rm{mod}~ 2\pi)$ if it has negative cosine value. 
We  summarize our findings  in the following theorem  which slightly generalizes similar results obtained in \cite{Lee,rogge}. 

\begin{thm}{All stable equilibria of a power grid with ring topology and homogeneous distribution of power consumption
and generation as described in Eq.~(\ref{eq:ALT}) are given by Eq.~(\ref{eqs:theta12}). 
Stability of the synchronous states in the network, corresponding 
to negative eigenvalues of the matrix $J$, is guaranteed as long as $|\delta_i-\delta_{i-1}|\,\leq\,\pi/2~(\rm{mod}~ 2\pi)$. If a
single phase difference exceeds $\pi/2~(\rm{mod}~ 2\pi)$ this synchronous state turns unstable and the corresponding equilibrium is type-1. 
Moreover, 
synchronized states with more than one absolute phase difference exceeding $\pi/2$ $(\rm{mod}~ 2\pi)$ correspond to equilibria with at
least two unstable directions, that is, to type-$j$ equilibria with $j>1$.}
\label{thm:stability}
\end{thm}
Since one positive eigenvalue of $J$ corresponds to one eigenvalue with positive real part of $L$, we only
need to analyze the eigenvalues of $L$. The proof follows after Theorem \ref{thm:matrxi_L} in Appendix A.  

Theorem \ref{thm:stability} confirms that Eqs.~(\ref{eqs:theta12}) indeed capture
all the stable equilibria of the homogeneous model. 

Before considering the case of heterogeneous power generation and consumption, we make two remarks.

{\bf Remark I} We notice that for the case $N{\equiv}0~({\rm mod~4})$
an infinite number of equilibria exist for the homogeneous model. We will not consider this  nongeneric
case here, but refer to the work of De Ville~\cite{Lee} for more details about this case.

{\bf Remark II}
The equilibria we found depend on $m$. For practical purposes the case $m=0$ 
is most desirable for transport of electricity, as in this case direct transport of power from the generator to the consumer is realized.
Direct transport from generator to consumer minimizes energy losses that always accompany the transport of electrical power. 
Only when $m=0$, the power is transported to the consumer directly. The power is transported clockwise if $m<0$ and counterclockwise 
if $m>0$ as shown in Fig. \ref{fig:power_transport}(a).

\begin{figure}[ht]
\centering
\includegraphics[scale=1.0]{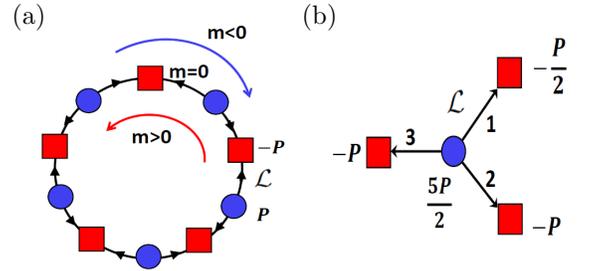}
\caption{{\small (a). A cyclic power grid with alternating consumers and generators, which may have stable equilibria with the power transported around
the cycle clockwise with $m<0$ and counterclockwise with $m>0$. The practical synchronization state is the one with $m=0$,  
in which the line load $\mathcal{L}=P/2$. (b). A tree power grid with 3 nodes. The line load of line 1 $\mathcal{L}=P/2$.}}
\label{fig:power_transport}
\end{figure}

For the case $m=0$, the stable equilibrium is $${\bm \theta}_S=(\theta_{1},\theta_{2},\cdots,\theta_{N})$$ with
$\theta_{1}=-\theta_{2}=\arcsin\left[\frac{P}{2K}\right]$ as follows from Eq.(\ref{eqs:theta12}).  
It is interesting to explore the ramifications of our results for the eigenvalues of $L$ of the second-order model. 
We write the eigenvalues of the matrix $L$ that
result after linearizing around the stable state (\ref{eqs:theta12}) with $m=0$, which can easily be determined:
\begin{align}
\hspace{-25pt}
\lambda_n&=-2\sqrt{4K^2-P^2}\sin^2\left(\frac{\pi n}{N}\right),~~~{n=0,1,\cdots,N-1}
\label{eq:lai}
\end{align}
The first nonzero eigenvalue, $$\lambda_1=-2\sqrt{4K^2-P^2}\sin^2(\pi/N),$$ gives rise to an associated
eigenvalue pair for matrix $J$
\begin{align}
\mu_{1,+}&=\frac{-\alpha}{2}+\frac{\sqrt{\alpha^2-8\sin^{2}{(\pi/N)}\sqrt{4K^2-P^2}}}{2},\nonumber\\
\mu_{1,-}&=\frac{-\alpha}{2}-\frac{\sqrt{\alpha^2-8\sin^{2}{(\pi/N)}\sqrt{4K^2-P^2}}}{2},
\end{align}
whose optimal value is obtained if $\alpha$ is tuned to the value which makes the square root vanish \cite{motter}.
For this value of $\alpha$, $\mu_{1,+}$=$\mu_{1,-}$=$\mu_{\rm opt}$, which equals
\begin{align}
\mu_{\rm opt}&=-\sqrt{-\lambda_1}=-\left(4K^2-P^2\right)^{1/4}\sqrt{2}\sin\left(\pi/N\right).
\label{eq:muopt}
\end{align}
From Eq.~(\ref{eq:muopt}) we easily observe that $\mu_{\rm opt}$ increases to $0$ with a rate of $1/N$ for $N$
sufficiently large. This suggests that networks with many nodes will be much more
susceptible to perturbations and hence it will be more difficult for the power grid to remain synchronized.

\subsection{The linear stability of heterogeneous model}

To investigate the effect of more heterogeneous distribution of power generation and consumption, we
determine the linear stability of the stable equilibria found using the numerical algorithm
described in Appendix D. We perform MC simulations to generate heterogeneous distributions of
power generation and consumption using the method given in Eq.~(\ref{eq:pres}), and average over 1000 runs.
In all runs we set $P=1$ and $K=8$, so $P/K=0.125$. In Fig.~\ref{fig:eigen}(a) we plotted the value of $-\mu_{\rm opt}$ for two 
values of $\sigma$ as a function of $N$. Indeed the dependence on $N$ is as predicted, and 
the two curves almost coincide, which means the eigenvalue is not so sensitive to the heterogeneity of power
distribution for the setting of $P$ and $K$. In Fig.~\ref{fig:eigen}(b) we explore 
the dependence on $\sigma$. Here we see as the heterogeneity of $P_{i}$ increases, the expected 
linear stability decreases. However only a very mild dependence
on $\sigma$ can be seen, so the heterogeneity does not seem to be very important for this value of $P/K$.
To better understand how each configuration of consumers and generators rather than the averaged configuration
changes its stability with increasing heterogeneity, we plotted the distribution
of $-\mu_{\rm opt}$ in Fig.~\ref{fig:eigen}(c) and (d). These show that besides a small shift of the maximum
toward smaller values of $-\mu_{\rm opt}$ the distribution is also broader, which indicates that
certain configuration will be less stable than others. We remark that 
the value of $y$ axis is relatively large, which means that the $-\mu_{\rm opt}$ is very close
to the average value. 

\begin{figure}[ht]
\centering
%
\includegraphics[scale=1.1]{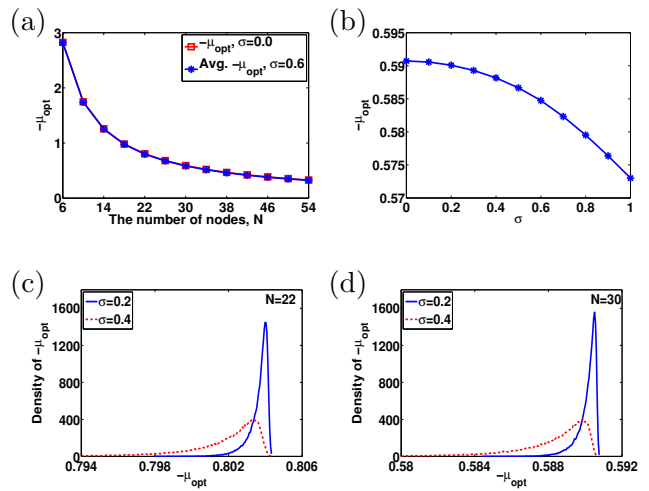}
\caption{{\small (a). $-\mu_{\rm opt}$ as a function of $N$ for $\sigma=0$ and $\sigma=0.6$. 
 (b). $-\mu_{\rm opt}$ as a function of $\sigma$ with $N=30$.
(c). The distribution of $-\mu_{\rm opt}$ for $\sigma=0.2$ and $\sigma=0.4$ and $N=22$.  (d).
The density of $-\mu_{\rm opt}$ for $\sigma=0.2$ and $\sigma=0.4$ where $N=30$. $P/K=0.125$ is kept fixed in all panels. }}
\label{fig:eigen}
\end{figure}
\section{Nonlinear stability of the synchronous state in ring networks}
\label{Sec:nonlinear}

We next discuss the stability of synchronous operation when the system is subject to perturbations
of such a degree that render the linear stability analysis of the previous section inappropriate. A measure for the
stability of the stationary states is then provided by the basin of attraction of the equilibria. For high-dimensional
systems this is a daunting task. However, it is possible to estimate the volume of the basin either
by numerical techniques, such as for example, the recently introduced basin stability $S$, by Menck \emph{et al.}\cite{menck,menck2},
in which the phase space is divided into small volumes. Choosing initial conditions in each of the small volumes
and recording convergence to a stable equilibrium  for each attempted initial condition, a number $S$ between $0$ and $1$ which
is a measure for the size of the volume of the attracting phase space, can be obtained.
Since this technique is computationally demanding and also labels solutions which make large
excursions through phase space as stable \cite{kurths2}, as they do belong to the stable manifold of the equilibrium, 
we will follow a different approach. 

The stability region has been analyzed by Chiang \cite{hirsch1} and independently Zaborsky \emph{ et al.} \cite{zab1,zab2} 
and the direct method was developed by Varaiya, Wu, Chiang \emph{et al.} \cite{Varaiya,book}
to find a conservative approximation to the basin of stability. 

We define an energy function $E(\bm{\delta},\bm{\omega})$ by
\begin{align}
\hspace{-25pt}
E(\bm{\delta},\bm{\omega})&=\frac{1}{2}\sum_{i=1}^N{\omega_i^2}-\sum_{i=1}^N P_i\delta_i-K \sum_{i=1}^N(\cos(\delta_{i+1}-\delta_i))\nonumber\\
\hspace{-25pt}&=\frac{1}{2}\sum_{i=1}^N{\omega_i^2}+V(\bm{\delta}),
\label{eq:energy}
\end{align}
where we defined the potential $V(\bm{\delta})$ as
\begin{align}
V(\bm{\delta})=-K\sum_{i=1}^N \cos(\delta_{i+1}-\delta_i)-\sum_{i=1}^N P_i \delta_i.
\label{def:pot}
\end{align}
It can easily be shown that
\[
\frac{d E(\bm{\delta},\bm{\omega})}{dt}=-\alpha\sum_{i=1}^N\omega_i^2\,\,{\leq}\,\,0.
\]

The primary idea behind estimating the region of attraction of a stable equilibrium by the direct method, is that
this region is bounded by a manifold $\cal{M}$ of the type-1 equilibria that 
reside on the potential energy boundary surface (PEBS) of the stable equilibrium.
The PEBS can be viewed as the stability boundary of the associated gradient system \cite{chiang2,Varaiya}
\begin{align}
\frac{d\delta_i}{dt}=-\frac{\partial V(\bm{\delta})}{\partial \delta_i}.
\label{eq:reduced}
\end{align}
The \emph{closest equilibrium} is 
defined as the one with the lowest potential energy on the PEBS. 
By calculating the closest equilibrium with potential energy $V_{\rm min}$ and equating this to the total energy,
it is guaranteed that points
within the region bounded by the manifold ${\cal{M}}=\{(\delta,\omega)|E(\delta,\omega)=V_{\rm min}\}$, will
always converge to the stable equilibrium point contained in ${\cal{M}}$. 

The idea of estimating  the region of stability by type-1 equilibria is probably best illustrated
by considering a simple example of a three-node network depicted in Fi-g.~\ref{fig:contour_3nodes}(a). We choose this network only 
for illustration purposes as this small three-node network allows direct evaluation.
 For this network  we set
$$P_1/K=0.125,P_2=-0.125,P_3/K=0$$ and $\alpha=0$. Equipotential curves were plotted in Fig.~\ref{fig:contour_3nodes}(b). 
The type-1 equilibria (saddles) are displayed as little circles and squares, numbered 1 to 6. It is clear that the type-1 equilibria indeed 
surround the stable equilibria which are
shown as local minima in the potential $V$. Equilibrium 1 is the closest equilibrium with the smallest potential
energy on the PEBS plotted
by a black dash-dotted line. A small
perturbation in the direction to saddle point 1, depicted by the red dashed curve 
leads to desynchronization, whereas a larger 
perturbation in a different direction (blue solid curve) eventually decays toward the
stable equilibrium point and hence the system
stays synchronized. This shows the conservativity of the direct method and the challenges in calculating the
region of stability, as it depends on both the direction and size of the perturbation. One approach to this
problem is to determine the so-called controlling unstable equilibrium point, which was developed by
Chiang \emph{et al.}\cite{hirsch1,chiang1}.
We will not consider this method here, but rather restrict ourselves to the potential energy of all the type-1 saddles on 
the PEBS. As displayed in Fig. \ref{fig:contour_3nodes}(b), 
there are two type-1 saddles corresponding to the absolute value $|\theta_{i}|$  
of a phase difference exceeding $\pi/2$. The potential energy of these two saddles are 
different and the one with smaller potential energy
is more susceptible
to perturbations. In the following study, all the equilibria are divided into two groups: (I) a group that corresponds 
to the phase difference $\theta_{i}$ exceeding $\pi/2$ with smaller energy and (II) the other group with larger energy. 
In Fig. \ref{fig:contour_3nodes}(b),
direct calculation shows that the saddles (1-3) constitute group I and (4-6) constitute group II. 

We remark that closest equilibrium 1 corresponds to the line connecting node 1 and 2 with the largest line load. 
This makes sense since the line with higher line 
load is easier to lose synchronization first. 

In subsection {\bf A}, we derive the analytical approximation of the potential energy of the equilibria on PEBS of
homogeneous model, for group I and group II respectively. In subsection {\bf B}, we present the numerical results
for the heterogeneous model.

\begin{figure}[ht]
\centering
\includegraphics[scale=1.1]{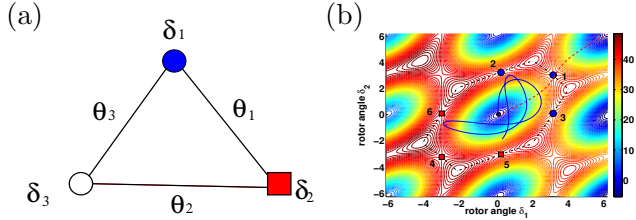}
\caption{{\small (a). A 3-node power grid. (b). The potential energy of the three nodes power grid as a function of $\delta_{i}$ 
where $P_{1}/K=0.125$, $P_{2}/K=-0.125$, and $P_{3}/K=0$. The 6 unstable equilibria are
local minima on the potential energy boundary surface (PEBS) plotted by the black dash-dotted line. 
The equilibrium 1 and 4, 2 and 5, 3 and 6 are caused by $\theta_{1}$, $\theta_{2}$ and $\theta_{3}$ exceeding $\pi/2$ respectively.  
Equilibrium 1 is
the closest equilibrium on the PEBS. The trajectory plotted by red dashed line
goes through saddle 3 and results in desynchronization after $\theta_{1}$ exceeded $\pi/2$. However, 
the trajectory plotted by the blue solid line always stays inside the
attraction of the stable equilibrium even though its energy is larger than the potential energy of saddle 1.}}
\label{fig:contour_3nodes}
\end{figure}

\subsection{Potential energy for homogeneous model}\label{subsubsec:analytical}

For the case of an $N$-node alternating ring network with $P_{i}$ distributed according to Eq.(\ref{eq:ALT}), we can easily find analytical expressions for the
potential energy by combining expressions for the equilibria (\ref{eqs:theta12}), the potential energy (\ref{def:pot}) and Eq.(\ref{eq:ALT}).
We assume that we are in a stable state, ${\bm \theta}_S^m=(\theta_1^m,\theta_2^m,\dots,\theta_1^m,\theta_2^m)$, that is,
$\theta_1^m,\theta_2^m\in[-\pi/2,\pi/2]$, which can always be achieved by a proper choice of $m$. For example $m=0$
corresponds to such a stable state.  The potential energy of the stable state ${\bm \theta}_S^m$ is
\begin{align}
\hspace{-15pt}
V_S^m&=-K\sum_{i=1}^N\cos(\theta_i^m)-\sum_{i=1}^NP(-1)^{i+1}\delta_i\nonumber\\
&=-\frac{KN}{2}\left[\cos(\theta_1^m)+\cos(\theta_2^m)\right]+\frac{NP}{4}\left[\theta_2^m-\theta_1^m\right].
\label{eq:p1}
\end{align}

We next consider the potential energy of the type-1 saddle points. According to Theorem \ref{thm:stability}, a type-1 saddle point 
corresponds to a link with absolute phase difference exceeding $\pi/2~(\rm{mod}~ 2\pi)$ in the network. 
We denote the type-1 saddle points
corresponding to the stable state ${\bm \theta}_S^m$ by 
$${\bm T}_j^m=(\hat{\theta}_1^m,\hat{\theta}_2^m,\dots,\pi-\hat{\theta}_1^m,\hat{\theta}_2^m,\dots,\hat{\theta}_1^m,\hat{\theta}_2^m).$$
 and
$${\bm {\bar{T}}}_j^m=(\hat{\theta}_1^m,\hat{\theta}_2^m,\dots,-\pi-\hat{\theta}_1^m,\hat{\theta}_2^m,\dots,\hat{\theta}_1^m,\hat{\theta}_2^m)$$
where the phase difference $\theta_j$ exceeds $\pi/2~(\rm{mod}~ 2\pi)$ and $j$ is odd. 
These two equilibria belong to group I and group II respectively.

In the following, we only focus on the type-1 saddle ${\bm T}_j^m$, the same results can be obtained for ${\bm {\bar{T}}}_j^m$.  


The equations that determine the values of $\hat{\theta}_1^m$ and $\hat{\theta}_2^m$ are now (\ref{t1a})
(with $\hat{\theta}_i^m$ substituted for $\theta_i^m$) combined with
\begin{align}
\left(\frac{N}{2}-2\right)\hat{\theta}_1^m+\frac{N}{2}\hat{\theta}_2^m=(2m-1)\pi.
\label{eq:thetm}
\end{align}

Hence we find that the type-1 saddles are implicitly given as solutions of the following equation
{\footnotesize{
\begin{align}
\hspace{-25pt}
\sin\left(\hat{\theta}_1^m-\frac{(2m-1)\pi}{N}-\frac{2\hat{\theta}_1^m}{N}\right)=\frac{P}{2K\cos\left(\frac{(2m-1)\pi}{N}+\frac{2\hat{\theta}_1^m}{N}\right)},
\label{eq:sad}
\end{align}
}}
which admits a solution $\hat{\theta}_1^m\in[0,\frac{\pi}{2}]$ when $$P/2K<\cos((2m-1)\pi/N+(2m+2\pi)/(N(N-2))).$$  We next argue that the
type-1 saddles found in Eq.~(\ref{eq:sad}) lie on the PEBS which surrounds the
stable equilibrium ${\bm \theta}_S^m$. One could use the same arguments as 
previously invoked  by De Ville\cite{Lee}. In Appendix C, we provide a more general proof which is valid for different $m$.

We set $m=0$ for the reasons described in remark II in section \ref{sec:linear} and denote ${\bm \theta}_S^0$ and ${\bm T}_{j}^0$ by $
{\bm \theta}_S$ and ${\bm T}_{j}$ respectively. We remark that there
are $2N$ type-1 equilibria on the PEBS of ${\bm \theta}_{S}$ if $P/2K$ is sufficiently large and 
each line load is $P/2$ for the equilibrium ${\bm \theta}_S$ 
as Fig.~(\ref{fig:power_transport})(a) shows.

We proceed to calculate the potential energy differences (details given in 
Appendix B) between the stable state ${\bm \theta}_S$ and the
saddle ${\bm T}_j$ for $j$ odd, which we call $\Delta V_I$
\begin{align}
\hspace{-25pt}
\Delta V_I&=-\frac{KN}{2}\left[\cos\hat{\theta}_1^0-\cos\theta_1^0\right]-\frac{KN}{2}\left[\cos\hat{\theta}_2^0-\cos\theta_2^0\right]\nonumber\\
\hspace{-25pt}&+2K \cos\hat{\theta}_1^0-\frac{NP}{4}\left[\hat{\theta}_1^0-{\theta}_1^0\right]+\frac{NP}{4}\left[\hat{\theta}_2^0-{\theta}_2^0\right]\nonumber\\
\hspace{-25pt}&+(-\pi/2+\hat{\theta}_1^0) P.
\label{eq:VI}
\end{align}
We can recast Eq.~(\ref{eq:VI}), using Eq.~(\ref{t1a}), in the following form
\begin{align}
\hspace{-15pt}
\Delta V_I&=P\left(-\frac{\pi}{2}+\arcsin\frac{P}{2K}\right)+\sqrt{4K^2-P^2}+\Delta U_I,
\label{eq:dv1}
\end{align}
where $\Delta U_I$ can be proven positive and has the asymptotic form for large $N$
\begin{align}
\hspace{-25pt}
\Delta U_I &=\frac{1}{N}\left(\frac{\pi}{2}-\arcsin\frac{P}{2K}\right)^2\sqrt{4K^2-P^2}+O\left(N^{-2}\right).
\label{eq:du1}
\end{align}
For ${\bm {\bar{T}}}_j^0$, a similar calculation shows that the 
potential energy difference can be expressed as
\begin{align}
\hspace{-15pt}
\Delta V_{II}&=P\left(\frac{\pi}{2}+\arcsin\frac{P}{2K}\right)+\sqrt{4K^2-P^2}+\Delta U_{II},
\label{eq:dv2}
\end{align}
where $\Delta U_{II}$ can be proven positive and has the asymptotic form for large $N$
\begin{align}
\hspace{-25pt}
\Delta U_{II} &=\frac{1}{N}\left(\frac{\pi}{2}+\arcsin\frac{P}{2K}\right)^2\sqrt{4K^2-P^2}+O\left(N^{-2}\right).
\label{eq:du2}
\end{align}

We remark for the case $j$ is even, the derivation of the potential energy differences is analogous.

From the expression for the energy barriers $\Delta V_I$ and $\Delta V_{II}$, we can easily infer that 
as the line load $\mathcal{L}=P/2$ increases, $\Delta V_{1}$ decreases and $\Delta V_{II}$ increases. As mentioned before,
$\Delta V_{I}$ is more susceptible to disturbances. 

Furthermore, we can immediately draw
the conclusion that for large network sizes, $\Delta V_I$ and $\Delta V_{II}$ approach a limiting value
that depends only on $K$ and $P$ which can be observed in Fig.~\ref{fig:potential_heter}(a).
A direct calculation shows that the asymptotic limits correspond
exactly to a potential difference found for a tree network, which is sketched in Fig.~\ref{fig:power_transport}(b).
We remark that the line load of each line in the ring network and line 1 in the four nodes tree network are both $P/2$. 
Indeed, we find that for the line in a tree network with line load $P/2$, the energy leading it to desynchronization are 
$\Delta V_I^T$ and $\Delta V_{II}^T$ \cite{book}
\begin{eqnarray}
\hspace{-25pt}
 &&\Delta V_{I}^{T}=\frac{P}{2}\big(-\pi+2\arcsin{\frac{P}{2K}}\big)+\sqrt{4K^{2}-P^{2}},\nonumber\\
 \hspace{-25pt} &&\Delta V_{II}^{T}=\frac{P}{2}\big(\pi+2\arcsin{\frac{P}{2K}}\big)+\sqrt{4K^{2}-P^{2}}.
\label{eq:treepot}
\end{eqnarray}

Hence the energy barrier in Eq.~(\ref{eq:dv1}) and (\ref{eq:dv2}) can be explained in terms 
$\Delta V_{I}=\Delta V_{I}^{T}+\Delta U_{I}$ and $\Delta V_{II}=\Delta V_{II}^{T}+\Delta U_{II}$. 
As $\Delta U_I$ and $\Delta U_{II}$ are always positive, 
the energy needed to make a line lose synchronization (exceeding $\pi/2~(\rm{mod}~ 2\pi)$) is increased for the line 
in a ring network compared with in
a tree network. In other words,  
the line with line load $\mathcal{L}=P/2$ in the ring network is more 
robust than in a tree network. This permits the line in cycles to transport more power. 
A ring topology will result in an increased stability of the synchronous state compared to that of
a tree network. This effect is larger for smaller networks. This finding corroborates the
results by Menck \emph{et al.}~\cite{menck2}, who found decreased stability from dead-ends or, small trees in the network. 

In order to examine the robustness of our results, we next perform numerically studies on the networks with the random configuration of
consumers and generators as in Eq.(\ref{eq:pres}).

\begin{figure}[ht]
\centering
%
\includegraphics[scale=1.1]{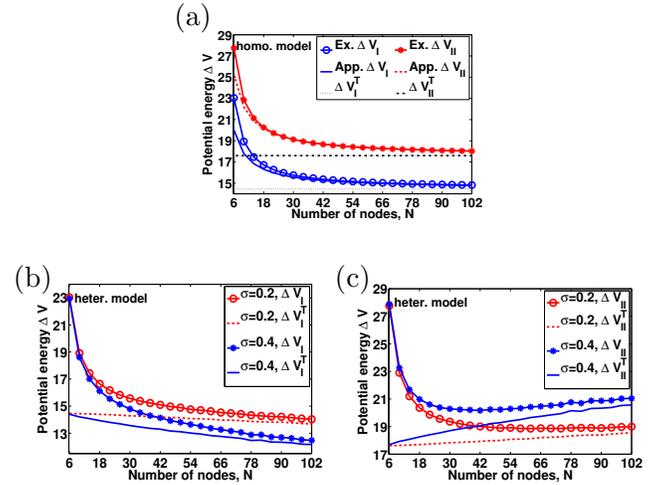}
\caption{{\small (a) The potential energy $\Delta V_{I}$, $\Delta V_{II}$, $\Delta V_{I}^{T}$ and
$\Delta V_{II}^{T}$ as functions of $N$ for the homogeneous model.
 The approximate value of $\Delta V_{I}$ and $\Delta V_{II}$ are calculated ignoring the term $O(N^{-2})$ in Eq. (\ref{eq:du1}) 
 and Eq.(\ref{eq:du2}) respectively.
 (b). The average value of $\Delta  V_{I}, \Delta V_{I}^{T}$ of heterogeneous model as
 functions of $N$ with $\sigma=0.2,0.4$. (c). The average value of $\Delta V_{II},\Delta  V_{II}^{T}$ of heterogeneous model as
 functions of $N$ with $\sigma=0.2,0.4$. $P=1$ and $K=8$ in all panels.}}
\label{fig:potential_heter}
\end{figure}

\begin{figure}[ht]
\centering
\includegraphics[scale=1.1]{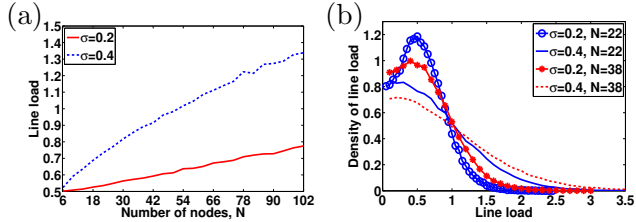}
\caption{{\small (a) The average line load as a function of $N$ for the heterogeneous model with $\sigma=0.2,0.4$.
(b) The distribution of line loads of cyclic power grids with $N=22, 38$ and $\sigma=0.2,0.4$. The distribution widens both
for increasing values of $N$ and $\sigma$. $P=1$ and $K=8$ in all panels.}}
\label{fig:line_load_dis}
\end{figure}

\subsection{Numerical results for the heterogeneous model}\label{subsubsec:numerical}

From the analysis of the nonlinear stability of cyclic power grids in homogeneous model, we know that 
the potential energy differences between the  type-1 equilibria
and the stable synchronous state with $m=0$, is always larger than the potential energy differences for a tree like network
with the same line load $\mathcal{L}$. Moreover, the potential energy barrier of the ring network approaches that 
of the tree network as $N$ increases. In the following, we verify whether this
remains true for cyclic power grids with heterogeneous distribution of $P_{i}$ and study how the heterogeneity
of power distribution influences the nonlinear stability.

\begin{figure}[ht]
\centering
%
%
\includegraphics[scale=1.1]{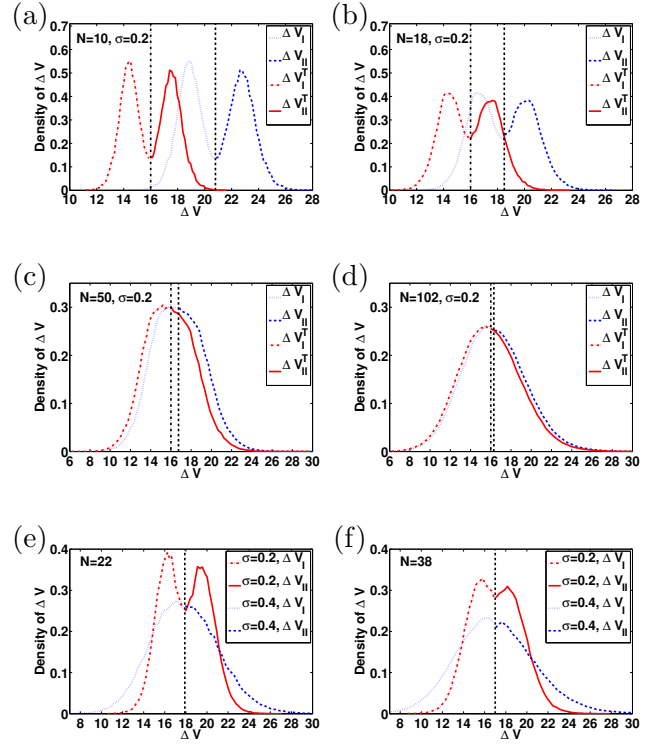}
\caption{{\small (a-d). The distribution of $\Delta V_{I}$, $\Delta V_{II}$, $\Delta V_{I}^{T}$ and $\Delta V_{II}^{T}$ of cyclic power grids
 for $N=10,18,50,102$. 
 (e-f). The distribution of $\Delta V_{I}$ and $\Delta V_{II}$ for $\sigma=0.2$, $0.4$ with $N=22, 38$.
 The black dash-dotted lines in the middle of figures (a-d) denote the boundary between $\Delta V_{I}^{T}$ and $\Delta V_{II}^{T}$ and 
 the black dashed lines in the middle of figures (a-f) indicate the boundary between $\Delta V_{I}$ and $\Delta V_{II}$. 
 $P=1$ and $K=8$ in all panels.}}
 \label{fig:potential_distribution}
\end{figure}


We next focus on how the potential energy of type-1 equilibria changes as $N$ increases. As we remarked in the previous subsection,
there are two groups of type-1 equilibria on the PEBS of ${\bm \theta}_{S}$, each having a
different potential energy relative to the synchronous state, $\Delta V_{I}$ and $\Delta V_{II}$, respectively.

As we do not have analytical expressions for $\Delta V_{I}$ and $\Delta V_{II}$ in this case,
we numerically compute these values for different values of $\sigma>0$ using the same procedure
for assigning values to $P_i$ as in Eq.~(\ref{eq:pres}).
For different values of $N$, we perform 2000 runs to calculate $\Delta V_{I}$ and $\Delta V_{II}$ and compute
the ensemble average. To determine which type-1 equilibria are on the PEBS of ${\bm  \theta}_{S}$, the numerical
algorithm proposed by Chiang \emph{et al.}~\cite{hirsch1} is used.

Since $\sigma$ is nonzero, incidentally a large value of $P_i$ can be assigned to a node, which prevents the existence of 
a stable equilibrium. Such runs will not be considered in the average.
Neither are runs in which fewer than  $2N$ type-1 equilibria are found on the PEBS.

In our numerical experiments, we  set again
$K=8, P=1, P/K=0.125$ and vary $N$ between 6 and 102 and set either  $\sigma=0.4$ or $\sigma=0.2$.


We determine the potential differences $\Delta V_{I}^T$ and $\Delta V_{II}^T$ by first calculating the
stable equilibria ${\bm {\theta}}_S$. As ${\bm {\theta}}_S$ determines all phase differences, it facilitates computing the line
loads between all connected nodes. From the line loads we subsequently extract the value
of $P$ which we then substitute into Eq.~(\ref{eq:treepot}) to find $\Delta V_{I}^{T}$ and
$\Delta V_{II}^{T}$, respectively.

By considering the average values of the quantities $\Delta V_{I}$, $\Delta V_{II}$, $\Delta V_{I}^{T}$, $\Delta V_{II}^{T}$,
we conclude the following.

First, for the heterogeneous distribution of $P_{i}$, 
the average value of $\Delta V_{I}$ and $\Delta V_{I}^{T}$ decreases with $N$ as shown in Fig.~\ref{fig:potential_heter}(b)
and Fig.~\ref{fig:potential_heter}(c). 
This is because the average line load increases with $N$ as shown in Fig.~\ref{fig:line_load_dis}(a) and 
$\Delta V_{I}$ and $\Delta V_{I}^{T}$ are monotonously increasing functions of the line load and $N$. 
However, $\Delta V_{II}$ decreases first and then 
increases after reaching a minimum with $N$ since it is a monotonously increasing function of line load but
a decreasing function of $N$. 
$\Delta V_{II}^T$ always increases since it is a monotonously increasing function of the line load. 

Second, for larger $\sigma$, $\Delta V_{I}$ decreases faster 
and $\Delta V_{II}$ increases faster after reaching a minimum. 
Since $\Delta V_{I}$ determines the stability more than $\Delta V_{II}$,  
the grid becomes less stable as $\sigma$ increases. So cyclic power grids with homogeneous distribution of $P_{i}$ as in 
Eq.~(\ref{eq:ALT}) are more stable than the ones with heterogeneous distributed $P_{i}$ as in Eq.~(\ref{eq:pres}). 

Third, $\Delta V_{I}$ and $\Delta V_{II}$ are always larger than $\Delta V_{I}^{T}$ and $\Delta V_{II}^{T}$, respectively and the
former two converge to the latter two
as $N$ increases, which is consistent with the homogeneous case. This confirms that the line in a cyclic grid is
more difficult to lose synchronization after a large perturbation
than in a tree grid. As the size $N$ of the cycle increases, this advantage disappears gradually. 

In order to get more insight in these scenario. The distribution of $\Delta V_{I}, \Delta V_{II}, \Delta V_{I}^{T}$ and $\Delta V_{II}^{T}$ 
are plotted in Figs.~\ref{fig:potential_distribution} for different $N$ and $\sigma$. 

The distribution of $\Delta V_{I}$ and
$V_{II}$ converge to $\Delta V_{I}^{T}$ and $\Delta V_{II}^{T}$, respectively, which can be observed
from Figs.~\ref{fig:potential_distribution}(a-d). There is a boundary between $\Delta V_{I}$ and
$\Delta V_{II}$ plotted by black dashed line in the middle of Figs.~\ref{fig:potential_distribution}(a-f). The boundary
actually is the upper bound of $\Delta V_{II}$ and lower bound of $\Delta V_{I}$, which is close to $2K+\frac{K\pi^{2}}{2N}$ calculated by
setting $P=0$ in Eqs.~(\ref{eq:dv1}) or (\ref{eq:dv2}). This does not depend on $\sigma$, as can be verified
in Figs.~\ref{fig:potential_distribution}(e-f).
For the tree connection, the boundary of $\Delta V_{I}^{T}$ and $\Delta V_{II}^{T}$ plotted by the black dash-dotted line  in
the middle of Figs.~\ref{fig:potential_distribution}(a-d) equals $2K$ calculated by setting $P=0$ in
Eqs.~(\ref{eq:treepot}).

Figs.~\ref{fig:potential_distribution}(e-f) show that the distribution of $\Delta V_{I}$ and
$\Delta V_{II}$ becomes broader as either $N$ or $\sigma$ increases.
This is
also reflected in the distribution of the line loads shown in Fig.~\ref{fig:line_load_dis}(b). We remark that
for the heterogeneous case, the line loads are different and the lines with
smaller line load become stronger while the ones with larger line load becomes weaker. In other word,
the power grid become stronger against some large disturbances while it becomes weaker against others.
As whatever $N$ or $\sigma$ increases, more lines become weaker which makes the network less stable against
various disturbances.

The maximum value of the density of potential energy is much smaller
than that of the linear stability as shown in Figs.~\ref{fig:eigen}(c-d). This demonstrates
that the potential energy is much more sensitive to the heterogeneity than the linear stability.

\section{Conclusion}\label{sec.conclusion}

Synchronization and their stability in cyclic power grids have been studied in this paper.
We obtained an analytical solution for the number
of stable equilibria of homogeneous cyclic power grids. The number of stable equilibria
increases linearly with the size of cyclic power grids. For cyclic power
grids with heterogeneous distribution of power generation and consumption, the existence of equilibria
has been analyzed with an efficient algorithm for finding all the type-1 equilibria.
Both the linear stability and nonlinear stability are investigated. 
Heterogeneity slightly reduces the linear stability, but affects the
nonlinear stability much more strongly.
We measure the nonlinear stability of the cyclic power grids
by the potential energy difference between the type-1 equilibria on PEBS and the stable equilibrium, 
which only depends on the power flow, but not on the damping coefficient.
An analytical approximation of the potential energy difference is obtained for the cyclic grids 
with a homogeneous distribution of generators and consumers. Numerical studies on the nonlinear stability
have been performed for the cyclic power 
grids with heterogeneous distribution of generators and consumers. For both the homogeneous and 
the heterogeneous case, we find that
the ring-like connection is more stable than the tree-like connection. A line connecting two nodes in a ring network is more robust than
a corresponding line in a tree network carrying the same line load, which allows it transport more power in the ring network. 
However, 
the greater stability of the ring configuration diminishes with a large network size. 
Therefore, to benefit from the increased stability of a ring like connection,
the network size should not
be too large (typically $N<10$).  

Compared to the homogeneous case, in heterogeneous cyclic power grids,
some lines become more stable while others become less stable since the line load becomes more heterogeneous. 
Hence the overall stability decreases. 

In real power grids, the stability of power grids can be enhanced by improving
the topology which is very complex. With this motivation, Kurths \emph{et al.}~\cite{menck2,kurths2} has explored the single node
basin stability and the survivability respectively to measure the nonlinear stability of power grids. The critical link capacity has also been
studied in refs\cite{Lozano2012,dorfler} to improve the topology. The Kuramoto order parameter \cite{timme2,Skardal2014,kurths}
also has been used to measure the synchrony of power grids. An analytical approximation of the critical clearing time
\cite{Anderson,Kundur} of faults 
in power systems is
derived by Roberts \emph{et al.}~\cite{Roberts} which shows that larger potential energy of the closest equilibrium may increase 
the critical clearing time. 
The potential energy of type-1 equilibria measures the energy-absorbing
capability of real power grids. Hence it can be used to measure the nonlinear stability as this paper presents.
The challenge is on how to find all the type-1 equilibria of the
power systems. There might be some other approach to approximate the potential energy such as
this paper presents namely line load. We present that lines transmitting the same amount of power may have
different stability similar to the difference between
tree like network and ring like network. It is worthwhile to investigate the nonlinear stability of small size artificial power grids to obtain
some insights on improving the stability measured by the potential energy of type-1 equilibria.

\section{Acknowledgement}
We thank  Jakob van de Woude for interesting conversations and comments during our regular meetings
and we are extremely grateful to Jan H. van Schuppen for his interest, good suggestions and invaluable mathematical help.

\appendix

\small{

\section{\large {Stable equilibria and type-1 saddles of cyclic power grids}}
 \begin{prop}
 The total number of stable equilibria in a ring network with homogeneous distribution of generation and consumption as in Eq.(\ref{eq:ALT})
is given
by  $N_s=1+2\lfloor \frac{N}{\pi}\arccos\left(\sqrt{\frac{P}{2K}}\right)\rfloor$.
 \end{prop}
\begin{proof}
To determine the number of  stable equilibria we require $\theta_1,\theta_2\in[-\pi/2,\pi/2]$.
Taking $$m\in H=\{-\lfloor N/2\rfloor,\dots,-1,0,1,\dots,\lfloor N/2\rfloor\}$$ and restricting to positive values of $m$ , 
we find the following inequality
$$\arcsin\left(\frac{P}{2K\cos(2 \pi m/N)}\right)\,\leq\,\frac{\pi}{2}-\frac{2 m \pi}{N}.$$
As $\sin x$ is a monotonic and positive function for $x\in[0,\pi/2]$, the inequality holds true
when taking $\sin$ on both sides. Using some trigonometry we arrive at the stated result, after
accounting for negative values of $m$ by multiplying with $2$ and adding $1$ to account for the $m=0$ term.
\end{proof}

\begin{thm}
The matrix $L$ has nonpositive eigenvalues if and only if  $|\delta_i-\delta_{i-1}|\,{\leq}\,\frac{\pi}{2}~(\rm {mod}~2\pi)$ for $i=1,\dots,N$.
A single phase difference $|\delta_{i}-\delta_{i-1}|\,{\geq}\,\pi/2 ~~(\rm {mod}~2\pi)$ will result
in one positive eigenvalue. If there are more than one phase 
differences $|\delta_{i}-\delta_{i-1}|\,{\geq}\,\pi/2~~(\rm {mod}~2\pi)$, the number of positive eigenvalues of $L$ is larger than 1. 
\label{thm:matrxi_L}
\end{thm}

\begin{proof}
The proof that all eigenvalues of $L$, are nonpositive can easily be established from Gershgorin's
circle theorem. We now prove that each phase difference exceeding $\pi/2~(\rm{mod}~ 2\pi)$ leads to
a positive eigenvalue. We will use matrix theory and Weyl's inequality to prove this, which
is in the same spirit as the proof of De Ville~\cite{Lee} in the case of a single frequency
Kuramoto network. We use the notation $\bm{\theta}^m_S$ to denote an \emph{stable} equilibrium with a fixed value
for $m$, chosen such that $\theta_1^m$, $\theta_2^m$ in the interval $[-\pi/2,\pi/2]$. The vector $\bm{\theta}^m_S$ reads in components
$(\theta_1^m,\theta_2^m,\dots,\theta_1^m,\theta_2^m)$. An equilibrium which has a single phase difference, say
between node $j$ and $j-1$, exceeding $\pi/2~(\rm{mod}~ 2\pi)$ will be denoted by
${\bf T}_j^m=(\theta_1^m,\theta_2^m,\dots,\pi-\theta_1^m,\dots\theta_1^m,\theta_2^m)$. Depending on the $j$ being odd or
even, a $\theta_1^m$ or a $\theta_2^m$ is replaced by $\pi-\theta_1^m$ or $\pi-\theta_2^m$, respectively.
As the system is rotationally symmetric we might as well choose the first phase difference to be larger than $\pi/2$.

In matrix language this implies that the Laplacian matrix $L$, which takes the following form in the case
with all phase differences  restricted to $[-\pi/2,\pi/2]$
{\footnotesize{
\begin{align}
\hspace{-25pt}
L=\left(\begin{array}{cccccc} -a-b & a & 0 & \dots& 0  & b \\
                            a & -a-b & b & 0& \dots & 0\\
                            0 & b & -a-b & a & 0 & \vdots \\
                            \vdots & \ddots & a & -a-b &  b & \vdots\\
                            0 &  \ddots & 0& b & -a-b & a \\
                            b & 0 & \dots & 0 & a & -a-b \end{array}\right),
\end{align}
}
}
with $a=K\cos(\theta_1^m)>0$ and $b=K\cos(\theta_2^m)>0$, will be transformed to the  Laplacian matrix $L'$:
{\footnotesize{
\begin{align}
\hspace{-25pt}
L'=\left(\begin{array}{cccccc} a-b & -a & 0 & \dots& 0  & b \\
                            -a & a-b & b & 0& \dots & 0\\
                            0 & b & -a-b & a & 0 & \vdots \\
                            \vdots & \ddots & a & -a-b &  b & \vdots\\
                            0 &  \ddots & 0& b & -a-b & a \\
                            b & 0 & \dots & 0 & a & -a-b \end{array}\right),
\end{align}}
}
when going from the equilibrium state $\bm{\theta}^m_S$ to the equilibrium state ${\bf T}_j^m$.
Both matrix $L$ and $L'$ have an eigenvalue $\lambda=\lambda'=0$, with eigenvector ${\bm{1}}$. Furthermore, all
other eigenvalues of $L$, which are real due to symmetry of $L$ are negative, and therefore
we can order the eigenvalues of $L$ as $\lambda_0=0\,{\geq}\,\lambda_1\,{\geq}\lambda_2\dots$.
As matrix $L'=L+{\bf x}^T {\bf x}$, with ${\bf x}=a/\sqrt{2} [1~ -1~ 0 ~\dots~ 0]$, we can use Weyl's
matrix inequality to relate the eigenvalues of $L'$: $\{\lambda'\}$ to that of $L$ as follows:
\begin{align}
\lambda'_0\,{\geq}\,\lambda_0=0\,{\geq}\,\lambda'_1\,{\geq}\,\lambda_1.
\end{align}
Hence, at most one eigenvalue of $L'$ is negative.
We can proof that one of the $\lambda'$ is negative by calculating the determinant of the reduced $(N-1)\times(N-1)$ matrix that
results after removing the eigenvalue $0$. $L'$ and $L'_{\rm \small red}$ have the same spectrum apart from $\lambda=0.$
{\scriptsize
\begin{align}
\hspace{-25pt}
L_{\rm\small red}'=\left(\begin{array}{cccccc} 2a-b & a+b &  & \dots&   &  \\
                            b & -a-b & b & 0& \dots & 0\\
                            0 & b & -a-b & a & 0 & \vdots \\
                            \vdots & \ddots & a & -a-b &  b & \vdots\\
                            0 &  \ddots & 0& b & -a-b & a \\
                            -b & -b & \dots & -b & a-b & -a-2b \end{array}\right),                            
\end{align}
}

We find, using recurrence relations, that the determinant of reduced $L'$ for general values of $N>4$ (even) is given  by
\begin{align}
&det(L'_{\rm red})=a^{N/2}a^{-1}b^{-1}b^{N/2}\times\nonumber\\
 &\left[\frac{3}{2} b N^2 + 20 b N + 56 b + 14 a N + 40 a + a N^2\right].\nonumber
\end{align}
Since the term in square brackets is always positive and $N$ is even in our case, this implies that the
determinant of the reduced matrix with dimensions $(N-1)\times (N-1)$ is positive and hence there will indeed
be exactly one positive eigenvalue. For the case $N\,{\leq}\,4$ direct calculation of the determinant gives the result.

For the case when there are two phase differences exceeding $\pi/2~(\rm{mod}~ 2\pi)$, the same argument can be applied, now starting from the matrix
$L'$ and defining a new matrix $L''$ in a similar fashion as above. Both the case that two neighboring phase differences
exceed $\pi/2~(\rm{mod}~ 2\pi)$ and the case that the two links with a phase difference ${\geq}\,\pi/2$ are separated by an odd number of 
nodes must be treated. Here we shall present the case for two non neighboring phase differences exceeding $\pi/2~(\rm{mod}~ 2\pi)$.
More precisely the equilibrium is $(\pi-\theta_1,\theta_2,\pi-\theta_1,\theta_2,\cdots,\theta_2)$. The matrix $L''$ is in this case
\begin{align}
\hspace{-25pt}
L''=\left(\begin{array}{cccccc} a-b & -a & 0 & \dots& 0  & b \\
                            -a & a-b & b & 0& \dots & 0\\
                            0 & b & a-b & -a & 0 & \vdots \\
                            \vdots & \ddots & -a & a-b &  b & \vdots\\
                            0 &  \ddots & 0& b & -a-b & a \\
                            b & 0 & \dots & 0 & a & -a-b \end{array}\right),
\end{align}
We remark again that the eigenvalues of $L'$ and $L''$ are related by  $L''=L'+{\bf x}^T {\bf x}$, with ${\bf x}=a/\sqrt{2} [0~ 0~ 1~ -1~ 0 ~\dots~ 0]$.
Using Weyl's matrix inequality we can proof that $L''$ will have 2 positive eigenvalues, when
the reduced $L''_{\rm\small red}$ matrix, which is constructed analogous to $L'_{\rm\small red}$, has a negative determinant. 
Using again recurrence relations, we can find $det(L''_{\rm \small red})$ and hence the product of the eigenvalues 
of $L''$ (with $\lambda=0$ removed)
\begin{align}
det~L''_{\rm red}&= -a^{N/2}b^{N/2}a^{-1}b^{-2}\left[(16 a^2 + 52 a b  - 12 b^2)\right. \\
&+\left. (-\frac{13}{2} a^2  - 20 ab + 3 b^2) N +
 \frac{3}{4} ( a^2 + 3 b a) N^2\right].\nonumber
\end{align}
This expression is valid for even $N$ and $N\,{\geq}\,4$, and is negative  definite, which shows that there will
be two positive eigenvalues in this case. For $N<4$, we can easily check our result by direct calculation.
The case with two neighboring nodes having phase differences exceeding $\pi/2~(\rm{mod}~ 2\pi)$ can be treated in a similar fashion.
In addition, it can be proven directly by 
Weyl's matrix inequality that $L$ has more than one positive eigenvalue if there are more than two phase differences
exceeding $\pi/2~(\rm {mod}~2\pi)$,
 \end{proof}

{\bf Remark} A beautiful proof using graph theory can be constructed as well, using the very general
approach of Bronski \emph{et al} \cite{bronski}. However, here we prefer direct calculation of the reduced determinant to show that a
direct calculation is feasible in this case.

\section{\large {Calculating the potential energy}}

Here we prove that the potential differences between the type 1-saddles $\bm{T}^m_j$ and stable equilibria $\bm{\theta}_S^m$ are
given by Eqs.~(\ref{eq:dv1}) and (\ref{eq:dv2})
with $\Delta U_{I}$ and $\Delta U_{II}$ positive and of the form Eq.~(\ref{eq:du1}) and (\ref{eq:du2}).

We will only consider $\Delta V_{I}$ as the proof for $\Delta V_{II}$ is analogous.
We start with rewriting Eq.~(\ref{eq:VI}) using the fact that we can write the linear term either $-PN\hat{\theta}_1/2$ 
or as $PN\hat{\theta}_2/2$. Choosing for the first form gives:
\begin{align}
\hspace{-25pt}
\Delta V_I&=-\frac{KN}{2}\left[\cos\hat{\theta}_1^m-\cos\theta_1^m\right]-\frac{KN}{2}\left[\cos\hat{\theta}_2^m-\cos\theta_2^m\right]\nonumber\\
&+2K \cos\hat{\theta}_1^m-\frac{NP}{4}\left[\hat{\theta}_1^m-{\theta}_1^m-\hat{\theta}_2^m+\theta_2^m\right]\nonumber\\
&+(-\pi/2+\hat{\theta}_1^m)P,
\label{ap1}
\end{align}

We next focus on the case $m=0$, as this is the state that we are most interested 
in and omit the index $m$ in the expressions. Moreover, we notice
that it is immediate from Eq.~(\ref{eq:sad}) and Eq.~(\ref{eq:thetm})  that
\begin{equation}
\hat{\theta}_1<\theta_1,~~~\hat{\theta}_2<\theta_2.
\end{equation}

We can recast Eq.~(\ref{ap1}) in the following form
\begin{align}
\Delta V_I&=\frac{KN}{2}[\cos\theta_1-\cos\hat{\theta}_1]+\frac{KN}{2}[\cos\theta_2-\cos\hat{\theta}_2]\nonumber\\
&+2K\cos\theta_1 +2K[\cos\hat{\theta}_1-\cos\theta_1]-\frac{NP}{4}[\hat{\theta}_1-\theta_1]\nonumber\\
&+P\left(\theta_1-\frac{\pi}{2}\right)+P(\hat{\theta}_1-\theta_1)+\frac{NP}{4}[\hat{\theta}_2-\theta_2].\nonumber\\
\label{ap2}
\end{align}
We next invoke the mean value theorem to write
\begin{align}
\cos\theta_1-\cos\hat{\theta}_1&=-(\theta_1-\hat{\theta}_1)\sin\xi_1,\nonumber\\
\cos\theta_2-\cos\hat{\theta}_2&=-(\theta_2-\hat{\theta}_1)\sin\xi_2,
\label{mv}
\end{align}
where $\xi_1\in[\hat{\theta}_1,\theta_1]$ and $\xi_2\in[\hat{\theta}_2,\theta_2]$
Realising that $\cos(\theta_1)=\sqrt{1-(P/2K)^2}$ and $\sin(\theta_1)=P/2K$, and using the mean value
expressions (\ref{mv}) we can rewrite Eq.~(\ref{ap2}) as
\begin{align}
\Delta V_{I}=\sqrt{4K^2-P^2}+\left(\arcsin(P/2K)-\frac{\pi}{2}\right)P+\Delta U_I,
\label{ap3}
\end{align}
where we introduced $\Delta U_I$, which is defined as
\begin{align}
\hspace{-25pt}
\Delta U_I&=-\frac{KN}{2}\left[\left(1-\frac{4}{N}\right)(\theta_1-\hat{\theta}_1)\sin\xi_1+(\theta_2-\hat{\theta_2})\sin\xi_2\right]\nonumber\\
\hspace{-25pt}&+\frac{K N}{2}\left[\left(\frac{P}{2K}\right)\left[\left(1-\frac{4}{N}\right)(\theta_1-\hat{\theta}_1)-(\theta_2-\hat{\theta}_2)\right]\right].
\end{align}
Assuming $N>4$, we can use the fact that $\sin$ is an increasing function on $[-\pi/2,\pi/2]$ to
show that
\begin{align}
\Delta U_{I}&\,{\geq}\,\frac{KN}{2}(\theta_1-\hat{\theta}_1)\left[\sin\theta_1-\frac{P}{2K}\right]\nonumber\\
&+\frac{KN}{2}(\theta_2-\hat{\theta}_2)\left[\sin\theta_2+\frac{P}{2K}\right]=0,
\label{eq:ap4}
\end{align}
where we used in the last step that $\theta_2=-\theta_1$ and $\theta_1=\arcsin(\frac{P}{2K})$.

Consider Eq.~(\ref{eq:sad}) for large $N$, when $N\to\infty$,  we
obtain Eq-s.~(\ref{eqs:theta12}) for the  stable equilibria. We therefore try to
find an approximate expression for type-1 saddles which are
separated from the stable equilibrium by a distance that is of order $1/N$.
We write $\hat{\theta}_1=\Delta\theta_1+\theta_1$ where $\Delta\theta_1$ is of order $1/N$. We can immediately
read off from Eq.~(\ref{eq:sad}) that
\begin{align}
\Delta\theta_1=\frac{1}{N}\left(-\pi+2\arcsin\frac{P}{2K}\right)+O(1/N^2).
\label{ap6}
\end{align}
From (\ref{ap6}) we find that $\Delta \theta_2=\Delta \theta_1$ is in the order $1/N$. By substituting
this in the expression for $\Delta U_I$ we find expression (\ref{eq:du1}). A similar analysis applies to $\Delta V_{II}$ and $\Delta U_{II}$

\section{\large {Proof of type-1 saddles being on the PEBS}}

In order to prove the type-1 saddles $\bm{T}_j^0$ are on the PEBS of $\bm{\theta}_S^{0}$, we 
need to prove that the unstable manifold of $\bm{T}_j^0$ limits
on $\bm{\theta}_S^{0}$. De Ville \cite{Lee} has provided a proof. Here we provide a more general 
proof which is valid for different $m$. By setting $P=0$, in Eq.~(\ref{eq:reduced}) and using
$\theta_i=\delta_i-\delta_{i-1}$, we find the equation of $\theta_i$.
\begin{eqnarray}\label{eq:nonlinear_P0}
&& \dot{\theta}_{i}=-2K\sin{\theta_{i}}+K\sin{\theta_{i-1}}+K\sin{\theta_{i+1}}.
\end{eqnarray}
Eq.~(\ref{eq:nonlinear_P0}) coincides with the overdamped limit of Eq.~(\ref{eq:stelsel}).

We start our proof by considering a particular stable equilibrium: $\bm{\theta}_S^m=(\theta_1^m,\theta_1^m,\cdots,\theta_1^m,\theta_1^m)$ with
$\theta_{1}^{m}=\frac{2\pi m}{N}$. 
There are $2N$ nearby saddles with coordinates given
by $$\bm{T}_j^m=(\hat{\theta}_1^m,\hat{\theta}_1^m,\cdots,\pi-\hat{\theta}_1^m,\cdots,\hat{\theta}_1^m,\hat{\theta}_1^m)$$ with
$\hat{\theta}_1^m=\frac{(2m-1)\pi}{N-2}$
and $$\bm{\bar{T}}_j^m=(\hat{\theta}_1^m,\hat{\theta}_1^m,\cdots,-\pi-\hat{\theta}_1^m,\cdots,\hat{\theta}_1^m,\hat{\theta}_1^m)$$ 
with $\hat{\theta}_1^m=\frac{(2m+1)\pi}{N-2}$,
which coincides with $\bm{T}_j^{m+1}$. We only prove $\bm{T}_j^m$ is on the PEBS of $\bm{\theta}_S^m$ and the proof of  $\bm{\bar{T}}_j^m$
is analogous.

We define the region 
{\scriptsize{$$I=\Big\{(\theta_{1}^{m},\cdots, \theta_{N}^{m})|\theta_{i}^m\in \left[\hat{\theta}_1^m,\pi-\hat{\theta}_1^m\right], i=1,\cdots N,
\sum_{i=1}^{N}\theta_{i}^{m}=2\pi m \Big\}$$}}
and show that $I$ is invariant under the dynamics
of Eq.~(\ref{eq:nonlinear_P0}). It can be easily verified that the stable equilibrium $\bm{\theta}_S^m$ 
and the saddles $\bm{T}_j^m$ are in the region $I$. 

We show that a trajectory $\bm{T}^{m}(t)=(\theta_1^m(t),\theta_2^m(t),\cdots,\theta_N^m(t))$ that starts in the region $I$ must remain in it. 
As in Eq.~(\ref{eq:additional_require}), equation $\sum_{i=1}^{N}\theta_{i}^{m}(t)=2\pi m$ always holds for $\bm{T}^{m}(t)$. If a component 
$\theta_{k}^m(t)$ goes through the upper bound 
with $\theta_{k}^m(t)\geq\pi-\hat{\theta}_1^m$, there must be at least one component 
crossing the lower bound since 
$\sum_{i=1}^{N}\theta_{i}^{m}(t)=2\pi m$. Otherwise, all the components satisfy $\theta_{i}^{m}(t)>\hat{\theta}_{1}^m, i\neq k$, which leads to
$\sum_{i}^{N}\theta_{i}^{m}(t)>2\pi m$, subsequently a contradiction results. So it is only 
needed to prove each component $\theta_{i}^m(t)$ can never go through the lower bound.

First, the component of $\bm{T}^{m}(t)$ can not go through the lower bound as a single component 
while two of its neighbors still remain in region $I$.
In fact, assume a component $\theta_{i}^{m}(t)$ is 
going through
its lower bound and its neighbors $\theta_{i+1}^{m}(t)$ and $\theta_{i-1}^{m}(t)$ are both in the range 
$(\hat{\theta}_1^m, \pi-\hat{\theta}_1^m)$. From Eq.~(\ref{eq:nonlinear_P0}), it yields

\begin{align}
 \dot{\theta}_{i}^m&=-2K\sin{\theta_{i}^m}+K\sin{\theta_{i+1}^m}+K\sin{\theta_{i-1}^m}\nonumber\\
 &>\,-2K\sin{\theta_{i}^{m}}
 +2K\sin{\hat{\theta}_{i}^{m}},
 \label{eq:lower_bound}
\end{align}
which is positive for $\theta_i^m(t)=\hat{\theta}_1^m$. Therefore $\theta_{i}^m$ will increase and thus never go beyond the lower bound.

Second, the components of $\bm{T}^{m}(t)$ can not go through the lower bound as a group of more than one neighboring component. 
Assume two components $\theta_{i}^{m}(t)$
and $\theta_{i+1}^{m}(t)$ are going through their lower bound and the neighbor $\theta_{i-1}^{m}(t)$ of $\theta_i^m(t)$
is in the range $(\hat{\theta}_1^m, \pi-\hat{\theta}_1^m)$, the inequality (\ref{eq:lower_bound})
for $\theta_{i}^{m}(t)$ must still be satisfied. 
It follows immediately $\theta_{i}^{m}(t)$ can not go through its lower bound. Similarly it holds for $\theta_{i+1}^{m}(t)$.

So the only possibility for the trajectory arriving the boundary of $I$ is all its components arrive together. However, only the 
type-1 saddles are on the boundary. Any trajectory starts from a point inside $I$ will always be 
inside $I$. So it has been proven that $I$ is an invariant set. Since $\bm{\theta}_S^m$ is the only stable equilibria in $I$, 
the unstable manifold of
$\bm{T}_j^{m}$ limits $\bm{\theta}_S^m$. In special, the unstable manifold of $\bm{T}_j^{0}(t)$ limits $\bm{\theta}_S^0$. 
In other words, the saddle $\bm{T}_j^{0}(t)$ is on the PEBS of $\bm{\theta}_S^0$.


This proof remains valid for sufficiently small values of $P/K$ as the dynamical system Eq.(\ref{eq:reduced}) is structurally 
stable\cite{dynamic_book}. 
Numerical simulations suggest that the result is  in fact true  for all $P/2K<1$.

\section{\large {Algorithm for finding all the type-j equilibria}}

As described in Theorem \ref{thm:stability}, if there is no phase difference with negative cosine value,
the related equilibria are stable. The solutions 
can be obtained directly by solving the following equations with different integer $m\in \left(-\frac{N}{4},\frac{N}{4}\right)$
  \begin{eqnarray}
  \hspace{-35pt}
  \sum_{i}^{N}\theta_{i}&=& 2m\pi, ~~ \theta_{i}\in [-\frac{\pi}{2},\frac{\pi}{2}],\label{eq:sum_theta}\\
  \hspace{-15pt} P_{i}-K\sin{\theta_{i}}+K\sin{\theta_{i+1}}&=&0, ~~~~i=1,\cdots,N-1. \nonumber
 \end{eqnarray}

In the following, we turn to the case with some phase differences having negative cosine values, i.e, phase differences exceeding 
$\pi/2~(\rm {mod}~2\pi)$

Assume the phase difference
$\theta_{i}$ to have a negative cosine value, the power flow  is $s_{i}=\sin{\theta_{i}}$. There must be a $\theta_{i}^{*}=\arcsin{s_{i}}$
such that $\theta_{i}=-\pi-\theta_{i}^{*}$ which can be plug in Eq.~(\ref{eq:sum_theta}) directly, the coefficient $\theta_{i}$ in equation should 
be changed by adding 
$\pi$ to the right hand side. Assume the first $i$ edges have negative cosine values. The problem is then reduced to  solving the equations
{\footnotesize
\begin{eqnarray}
-\theta_{1}-\cdots-\theta_{i}+\theta_{i+1}+\cdots+\theta_{N}=m\pi,  \theta_{i}\in [-\frac{\pi}{2},\frac{\pi}{2}].\label{eq:appendix_nonlinear1}\\
 P_{i}-K\sin{\theta_{i}}+K\sin{\theta_{i+1}}=0, ~~~~i=1,\cdots,N-1,
   \label{eq:appendix_nonlinear2}
\end{eqnarray}
}
where $m\in\left(-\frac{N}{2},\frac{N}{2}\right)$ is an even integer for even $i$ and is an odd integer for odd $i$.

Now let us focus on solving the nonlinear equations Eqs.-(\ref{eq:appendix_nonlinear1},\ref{eq:appendix_nonlinear2}).
Assuming that $\beta=\sin{\theta_{N}}$ is known, all the other variables can be
solved as follows
\begin{eqnarray}
\theta_{i}=\arcsin\Big(\frac{\sum_{j=1}^{i}P_{j}}{K}+\beta\Big),i=1,\cdots,N-1.
\label{eq:appendix_nonlinear3}
\end{eqnarray}
By substituting the values for $\theta_i$ from Eq.~(\ref{eq:appendix_nonlinear1}), we arrive at 
a one dimensional nonlinear equation for $\beta=\sin{\theta_{N}}$
\begin{eqnarray}
 \sum_{i=1}^{N}a_{i}\arcsin\Big(\frac{\sum_{j=1}^{i}P_{j}}{K}+\beta\Big)-m\pi=0.
 \label{eq:appendix_nonlinear4}
\end{eqnarray}
Where the coefficient $a_{i}=1$ if edge $i$ has phase difference $\theta_i$ with positive cosine value 
otherwise $a_{i}=-1$. All solutions to Eq.~(\ref{eq:appendix_nonlinear4})
can easily be found using the interval method described in \cite{moore}.

Herein, we denote the number of phase differences leading to  negative cosine values by $N_{n}$. Following the result on the spectral 
of weighted graph in Ref. \cite{bronski}, to find all the type-$j$ equilibria, it is only needed to solve the equilibria with $N_{n}=j-1,j,j+1$. 

We construct the following algorithm to find all the type-$j$ equilibria of the cyclic power grids.

%
%
\begin{algorithm}{ \emph {Finding all the type-$j$ equilibria of cyclic power grids}}
\item
\begin{algorithmic}[1]
\algsetup{indent=1em}
\FOR{$N_{n}=j-1,j,j+1$}
\STATE Collect all the combinations with $N_{n}$ lines whose phase difference have negative cosine values and determine
the coefficient $a_{i}$ in Eq.~(\ref{eq:appendix_nonlinear4}). Store all the combinations and corresponding $a_{i}$ in the set $\mathcal{C}$.
\IF {$N_{n}$ is an odd integer}
\STATE Collect odd numbers of $m$ in $(-\frac{N}{2},\frac{N}{2})$ and store in the set $\mathcal{S}_{m}$,
\ELSE 
\STATE Collect even numbers of $m$ in $(-\frac{N}{2},\frac{N}{2})$ and store in the set $\mathcal{S}_{m}$. 
\ENDIF
\FOR{all the combinations in $\mathcal{C}$}
\FOR{all $m$ in $\mathcal{S}_{m}$}
\STATE Solve the one dimensional Eq.~(\ref{eq:appendix_nonlinear4}) by interval method;
\IF{the solution $\beta$ exists}
\STATE Calculate all the phase difference $\theta_{i}$ by Eq.(\ref{eq:appendix_nonlinear3}). The ones with negative cosine 
values are $\pi-\theta_{i}$.
\STATE Adjust $\theta_{N}$ such that $\sum_{i=1}^{N}\theta_{i}=0$.
\IF{the equilibrum is type-j}
\STATE Save the equilibrium.
\ENDIF
\ENDIF
\ENDFOR
\ENDFOR
\ENDFOR
\end{algorithmic}
\end{algorithm}

%
%
%
%
%

Note that
the computation complexity for finding all the equilibria will definitely increase at least exponentially with the number of nodes $N$.
However, on finding all
the type-1 equilibria, we have $N_{n}\leq 2$. The computational
cost on solving all the type-1 equilibria is therefore $O(N^{3})$. We remark that 
this algorithm can be extended to networks without so many cycles.  

}

\small{

}

\end{document}